\documentclass[8pt]{llncs}
\usepackage{geometry} 
\usepackage{amsfonts,amssymb,amstext,amsmath}
\usepackage{color}

\newcommand{\adan}[1]{#1}

\title{Hidden variables \adan{simulating} quantum contextuality increasingly violate the Holevo bound}

\author{Ad\'an Cabello \inst{1} \and Joost J. Joosten \inst{2}}

\institute{
Departamento de F\'{\i}sica Aplicada II\\
Universidad de Sevilla, 41012 Sevilla, Spain\\
{\tt adan@us.es}\\
www.adancabello.com \\
\ \\
\and
Dept. L\`ogica, Hist\`oria i Filosofia
de la Ci\`encia \\  Universitat de Barcelona,
Montalegre 6,
08001 Barcelona, Spain\\
{\tt jjoosten@ub.edu}\\
www.phil.uu.nl/$\sim$jjoosten/}
\date{today} 

\begin{document}

\maketitle


\begin{abstract}
In this paper we approach some questions \adan{about quantum
contextuality} with tools from formal logic. In particular, we
consider an experiment associated with the Peres-Mermin square.
The language of all possible sequences of outcomes of the
experiment is classified in the Chomsky hierarchy and seen to
be a regular language.

\ \ \ \ \ Next, we make the rather evident observation that a
finite set of hidden finite valued variables can never account
for indeterminism in an ideally isolated repeatable experiment.
We see that, when the language of possible outcomes of the
experiment is regular, as is the case with the Peres-Mermin
square, the amount of binary-valued hidden variables needed to
de-randomize the model for all sequences of experiments up to
length $n$ grows as bad as it could be: linearly in $n$.

\ \ \ \ \ We introduce a very abstract model of machine that
simulates nature in a particular sense. A lower-bound on the
number of memory states of such machines is proved if they were
to simulate the experiment that corresponds to the Peres-Mermin
square. Moreover, the proof of this lower bound is seen to
scale to a certain generalization of the Peres-Mermin square.
For this scaled experiment it is seen that the Holevo bound is
violated and that the degree of violation increases uniformly.
\end{abstract}


\section{Introduction}


In this paper we will focus on an experiment that is associated
to the famous Peres-Mermin square \cite{Mermin,Peres}.
\adan{The experiment consists of a sequence of measurements
performed consecutively on a two-qubit system. All the
measurements are randomly chosen from a subset of those
represented by two-fold tensor products of the Pauli matrices
$X$, $Y$ and $Z$, and the identity $\mathbb{I}$.} The set
$\mathcal L$ of all sequences of \adan{outcomes consistent with
Quantum Mechanics} is studied as a formal language.

In the theory of formal languages, the Chomsky hierarchy
\cite{Chomsky1956,Sipser1997} defines a classification of
languages according to their level of complexity. In
Section~\ref{Section:ExperimentAndLanguage}, this language
$\mathcal L$ will be classified in the Chomsky hierarchy. It
will be seen to live in the lower regions of the hierarchy.
More concrete, it will be seen that the language is of Type 3,
also called \emph{regular}.

In
\adan{Subsection~\ref{subSection:HiddenVariablesAndDerandomization}},
the rather evident observation is made that a finite set of
hidden finite valued variables can never account for
indeterminism in an ideally isolated repeatable experiment.

Finally, in \adan{Subsection}~\ref{SectionHolevoBound}, the
question is addressed how much memory is needed to simulate
\adan{Quantum Mechanics in experiments with sequential
measurements}. The question naturally arises\adan{:} \emph{how}
are we allowed to simulate nature. We wish to refrain from
technical implementation details of these simulations as much
as possible. To this extent, we invoke the Church-Turing thesis
that captures and mathematically defines the intuitive notion
of what is computable at all by what mechanized and controlled
means so-ever. This gives rise to our notion of MAGAs:
Memory-factored Abstract Generating Automata. We prove a lower
bound for the amount of memory needed for MAGAs that simulate
extensions of the \adan{experiment associated to the
Peres-Mermin square} and shall see that this bound directly and
increasingly so violates the Holevo bound.


\section{Language defined by the experiment}\label{Section:ExperimentAndLanguage}


In this paper we will denote the Peres-Mermin square by the
following matrix
\[ \left( \begin{array}{ccc}
A & B & C \\
a & b & c \\
\alpha & \beta & \gamma \end{array} \right) \ \ \ \ (\dag),
\]
where these variables can get assigned values in $\{ 1, -1\}$.
The corresponding --classically impossible to satisfy--
restriction is that the product of any row or column should be
1 except for $C,c,\gamma$ which should multiply to $-1$.


Basically, contextuality amounts to the phenomenon that the
\adan{outcome of a measurement on a system relates to and
depends on other (compatible) measurements performed on that
system.}

In the next Subsection~\ref{Subsection:TheExperiment}, we will
first formally describe the possible outcomes of the experiment
that corresponds to the Peres-Mermin square. We will describe
this in almost tedious detail as we later need to formalize the
corresponding language.


\subsection{The experiment} \label{Subsection:TheExperiment}


We will collect the nine observables of our experiment into an
alphabet $\Sigma$ which we denote by
\[
\Sigma \ \ := \ \ \{ A, B, C, a , b, c, \alpha, \beta, \gamma\}.
\]
The experiment consists of arbitrarily many discrete
consecutive measurements of these nine observables which can
take values in the two-element set $\{ -1,1\}$. For reference
we reiterate that $(\dag)$ in this paper coincides with the
well-studied Peres-Mermin square \cite{Mermin,Peres}.

\begin{definition}[Context; Compatible observables]
The rows and columns of matrix $(\dag)$ are called
\emph{contexts}. Two observables within the same context are
called \emph{compatible} and two observables that do not share
a common context are called \emph{incompatible}.
\end{definition}

It is clear that each observable belongs to exactly two
contexts. Likewise, each observable is compatible with 4 other
observables and incompatible with 4 yet other observables. Now,
let us define what it means for an observable to be
\emph{determined}.

\begin{definition}[Determined observables; Value of a determined observable]\label{definition:experimentNineElements}
An observable becomes (or stays) {\bf determined} if:
\begin{itemize}
\item[(E1)] Once we measure an observable, it becomes (or
    stays) determined and its value is the value that is
    measured, either 1 or $-1$. If the observable $X$ that
    was measured was already determined, then the value of
    $X$ that is measured anew must be the same as in the
    most recent measurement.
\item[(E2)] If two observables within one context are
    determined, and if the third observable in this context
    was not yet determined, this third value becomes
    determined too. Its corresponding value is such that
    the product of the three determined values  in this
    context equals 1. The sole exception to this value
    assignment is the context $\{ C,c , \gamma \}$ that
    should multiply to $-1$.
\end{itemize}
The notion of an observable being \emph{undetermined} is defined by the following clause:
\begin{itemize}
\item[(D1)] By default, an observable is {\bf undetermined}
    and only becomes determined in virtue of (E1) or (E2).
    An observable $X$ that is determined remains determined
    {\bf if and only if} all successive measurements are in
    one of the two contexts of $X$ that is, all successive
    measurements are compatible with $X$. As soon as,
    according to this criterion, an observable is no longer
    determined we say that it has become {\bf
    undetermined}. Undetermined observables stay
    undetermined until they become determined. Undetermined
    observables have no value assigned. Sometimes we will
    say that the value of an undetermined observable is
    \emph{undefined}.
\end{itemize}
\end{definition}

A sequence of measurements is \emph{consistent} with our
experiment if its determined observables meet the restrictions
above. In essence, part of our definition is of inductive
nature. To see how this works, let us see this, by way of
example, the sequence of measurements $[A=1; B=1; c=1; \gamma
=1 ]$ is inconsistent with our experiment:

\medskip

\begin{tabular}{l|l|p{7cm}}
Measured & Measured\hspace*{5mm} & Comments\\
observable\hspace*{5mm} & value & \\
\hline
$A$ & 1 & The experiments starts so by default, all observables were undetermined (D1). After the measurement, by (E1) the observable $A$ is determined and assigned the value 1.\\
\hline
$B$ & 1 & As $B$ is a new measurement, the observable becomes determined (E1) with value 1. The observable $A$ remains determined by $(D1)$. Moreover, the observable $C$ which is in the context $\{A,B , C\}$, now becomes determined in virtue of (E2) with value 1 too as the product $A\cdot B \cdot C$ should multiply to 1.\\
\hline
$c$ & 1 & By (E1), $c$ becomes determined with value 1 and the observables $A$ and $B$ become undetermined in virtue of $(D1)$. The observable $C$ remains determined by (D1) with value 1 as the new measurement of $c$ is in the context $\{ C, c, \gamma \}$. Thus, in virtue of (E2), the observable $\gamma$ becomes determined too. Its value must be $-1$ as $C\cdot c \cdot \gamma$ should multiply to $-1$.\\
\hline
$\gamma$ & 1 & As we said before, in virtue of $(E2)$ the value of $\gamma$ should be $-1$. Thus this is inconsistent with the measurement of 1 for $\gamma$.\\
\end{tabular}\ \\
\ \\
Note that only the last measurement was inconsistent with our experiment.

\begin{remark}
If we fill the square $(\dag)$ with any assignment of $1$s and
$-1$s, then the number of products of contexts that equal $-1$
will always be even.
\end{remark}

\begin{proof}
By induction on the number of $-1$s. If all observables in
$(\dag)$ are set to one, then all contexts multiply to one thus
yielding zero --an even number-- of negative products. If we
add one more $-1$ to $(\dag)$, this $-1$ will occur in exactly
two contexts thereby flipping the sign of the respective
products of these two contexts.
\end{proof}


\subsection{The formal language}


In this subsection we shall specify a formal language
comprising exactly the possible strings of measurements of our
above specified experiment. We mention where the complexity of
this language resides in the Chomsky hierarchy: the regular
languages. Let us first briefly introduce some terms and
definitions from the theory of formal languages.

We shall call a collection of symbols an \emph{alphabet} and
commonly denote this by $\Sigma$. A \emph{string} or \emph{word} over an
alphabet is any finite sequence of elements of $\Sigma$ in
whatever order. We call the sequence of length zero the
\emph{empty string}, and will denote the empty string/word by
$\lambda$. We will denote the set of all strings over $\Sigma$
by $\Sigma^*$ using the so-called Kleene-star. Thus, formally
and without recurring to the notion of sequence, we can define
$\Sigma^*$, the set of all finite strings over the alphabet $\Sigma$
as follows.
\[
\begin{array}{c}
\lambda \in \Sigma^*;\\
\sigma \in \Sigma^* \ \& \ s\in \Sigma \ \Rightarrow \ \sigma s \in \Sigma^*.
\end{array}
\]
Instead of writing $\lambda \sigma$, we shall just write
$\sigma$. It is clear that $\Sigma^*$ is an inductive
definition so that we also have an induction principle to prove
or define properties over $\Sigma^*$. For example, we can now
formally define what it means to \emph{concatenate} -- stick
the one after the other-- two strings: We define $\star$ to be
the binary operation on $\Sigma^*$ by $\sigma \star \lambda =
\sigma$ and $\sigma \star (\tau s) = (\sigma \star \tau) s$.
Any subset of $\Sigma^*$ is called a \emph{language} over
$\Sigma$.

The study of formal languages concerns, among others, which
kind of grammars define which kind of languages, and by what
kind of machines these languages are recognized. In the current
paper we only need to provide a formal definition of so-called
\emph{regular} languages. We do this by employing regular
grammars. Basically such a grammar is a set of rules that tells
you how strings in the language can be generated.

\begin{definition}[Regular Grammar]
A regular grammar over an alphabet $\Sigma$ consist of a set
$\mathcal{G}$ of \emph{generating symbols} together with a set
of \emph{rules}. In this paper we shall refer to generating
symbols by using a line over the symbols. The generating
symbols always contain the special start-symbol $\overline{S}$.
Rules are of the form
\[
\begin{array}{ll}
\overline{X} \to \lambda & \mbox{or}\\
\overline{X} \to s \overline{Y},
\end{array}
\]
where $\overline{X}, \overline{Y} \in \mathcal{G}$ and $s\in
\Sigma$. The only restriction on the rules is, that there must
be at least one rule where the left-hand side is
$\overline{S}$.
\end{definition}

Informally, we state that a \emph{derivation} in a grammar is
given by repeatedly applying possible rules starting with
$\overline{S}$, where the rules can be applied within a
context. Thus, for example, when we apply the rule
$\overline{A} \to a \overline{B}$ in the context $\sigma A$, we
obtain $\sigma a \overline{B}$. A more detailed example of a
derivation is given immediately after
Definition~\ref{definition:regularLanguageNineElements}. We say
that a string $\sigma$ over $\Sigma$ is derivable within a
certain grammar if there is a derivation resulting in $\sigma$.
The \emph{language defined by a grammar} is the set of
derivable strings over $\Sigma^*$. A language is called
\emph{regular} if it is definable by a regular grammar.

We are now ready to give a definition of a regular language
that, as we shall see, exactly captures the outcomes of our
experiment. To this end, we must resort to a richer language
than just $\Sigma$ as $\Sigma$ only comprises the observables
and says nothing over the outcome\adan{s}. So, we shall
consider a language where, for example, $\tilde{A}$ will stand
for, ``$A$ was measured with value $-1$'', and $A$ will stand
for, ``$A$ was measured with value $1$''. We will denote this
alphabet by $\tilde \Sigma$. We will use the words
\emph{compatible} and \emph{incompatible} in a similar fashion
for $\tilde \Sigma$ as we did for our observables in $\Sigma$.
Thus, for example, we say that both $B$ and $\tilde B$ are
compatible with $C$. The only difference will be that $\tilde
A$ is not compatible with $A$ whereas $A$ is.

\begin{definition}[A grammar for $\mathcal{L}$]\label{definition:regularLanguageNineElements}
The language $\mathcal{L}$ will be a language over the alphabet
$\tilde \Sigma := \{  A, \tilde{A}, B, \tilde{B}, \ldots,
\beta, \tilde{\beta}, \gamma, \tilde{\gamma}  \}$, where the
intended reading of $A$ will be that the observable $A$ was
measured to be $1$, and $\tilde{A}$ will stand for measuring
$-1$, etc.

Before we specify the grammar that will generate $\mathcal{L}$,
we first need some notational conventions. In the sequel, $U,
V, X, Y$ and $Z$ will stand for possible elements of our
alphabet. If $X$ and $Y$ are compatible symbols, we will denote
by $Z(XY)$ the unique symbol that is  determined in (E2) by $X$
and $Y$. Thus, for example, $Z(A\tilde{B})= \tilde{C}$ and
$Z(C\gamma)=\tilde{c}$.

The generating symbols of the grammar will be denoted by a
string with a line over it. Note that, for example, $\overline{
XY}$ is regarded as one single generating symbol. The intended
reading of such a string is that the two symbols are different
and compatible, the last symbol is the one that can be
generated next, and the remainder of the string codifies the
relevant history. As usual we will denote the initial
generating symbol by $\overline{S}$. Let $\mathcal{L}$ be the
formal language generated by the following grammar.
\[
\begin{array}{llll}
\overline{S} &\ \ \ \  \longrightarrow \ \ \ \  & \lambda & \mbox{\ \ \ \ \ \ \ \ \   }\\

\overline{S} &\ \ \ \  \longrightarrow \ \ \ \  & \overline{X} & \mbox{\ \ \ \ \ \ \ \ \  for any symbol } X \in \Sigma \\

\ & \ & \ & \ \\
\overline{X} &\ \ \ \  \longrightarrow \ \ \ \  & X &  \\

\overline{X} &\ \ \ \  \longrightarrow \ \ \ \  & X\overline{ X} &   \\

\overline{X} &\ \ \ \  \longrightarrow \ \ \ \  & X\overline{ Z } & \mbox{\ \ \ \ \ \ \ \ \  for $Z$ incompatible with } X  \\

\overline{X} &\ \ \ \  \longrightarrow \ \ \ \  & X\overline{ XY} & \mbox{\ \ \ \ \ \ \ \ \  for $Y$ compatible with $X$ (but not equal)}   \\

\ & \ & \ & \ \\

\overline{XY} &\ \ \ \  \longrightarrow \ \ \ \  & Y &  \\

\overline{XY} &\ \ \ \  \longrightarrow \ \ \ \  & Y\overline{ XY} & \\

\overline{XY} &\ \ \ \  \longrightarrow \ \ \ \  & Y\overline{YX} & \\

\overline{XY} &\ \ \ \  \longrightarrow \ \ \ \  & Y \overline{ YZ} & \mbox{\ \ \ \ \ \ \  for $Z$ compatible with $Y$ (not equal),}\\

\ & \ & \ & \ \mbox{\ \ \ \ \ \ but not with } X \\

\overline{XY} &\ \ \ \  \longrightarrow \ \ \ \  & Y\overline{YZ(XY)} & 
\\

\overline{XY} &\ \ \ \  \longrightarrow \ \ \ \  & Y\overline{Z(XY)U} & \mbox{\ \ \ \ \ \ \ \ \  for $U$ compatible with $Z(XY)$ (but not equal)}\\

\ & \ & \ & \ \mbox{\ \ \ \ \ \ \ \ but not compatible with $X$ or } Y\\

\end{array}
\]
\end{definition}
We emphasize that the conditions on the right are not part of the rules. Rather, they indicate how many rules of this type are included in the grammar. For example, $\overline{S}  \longrightarrow    \overline{X}  \mbox{\ for any symbol } X {\in} \Sigma$
is our short-hand notation for nine rules of this kind.

Let us give an example of how this grammar works to the effect
that $ABc\tilde{\gamma}$ is in our language. Recall our reading
convention that says that $A$ stands for measuring $A=1$, $B$
for measuring $B=1$, $c$ for $c=1$, and $\tilde{\gamma}$ for
measuring $\gamma = -1$. Here goes a derivation of the string
$ABc\tilde{\gamma}$:
\[
\begin{array}{c|l|l}
\mbox{{\bf String derived\ \ }} & \mbox{{\bf Instantiation of rule\ \ }} & \mbox{{\bf General rule applied}}\\
\hline
\vspace{-0.4cm }
\ & \ & \ \\
\overline{A} & \overline{S} \longrightarrow \overline{A} &\mbox{By the rule $\overline{S} \longrightarrow \overline{X} $ with $X= A$} \\

\hline
\vspace{-0.4cm }
\ & \ & \ \\
A\overline{AB} & \overline{A} \longrightarrow A\overline{AB}  &\mbox{By the rule $\overline{X} \longrightarrow X\overline{XY}$ ($X$ and $Y$ compatible)}\\
\ & \ & \mbox{with $X= A$, $Y= B$. Note that $A$ and $B$ are}\\
\ & \ & \mbox{indeed compatible.} \\

\hline
\vspace{-0.4cm }
\ & \ & \ \\

AB\overline{Cc} & \overline{AB} \longrightarrow B\overline{Cc}  &\mbox{By the rule $\overline{XY} \longrightarrow  Y\overline{Z(XY)U}$  for $U$ compatible}\\
\ &\ &\mbox{with $Z(XY)$ (but not equal), but not compatible}\\
\ &\ &\mbox{with $X$ or $Y$, where $X= A$, $Y= B$, $Z(XY) = C$} \\
\ &\ &\mbox{and $U = c$.} \\
\hline
\vspace{-0.4cm }
\ & \ & \ \\
ABc\overline{c\tilde{\gamma}} & \overline{Cc} \longrightarrow c\overline{c\tilde{\gamma}}  &\mbox{By the rule $\overline{XY} \longrightarrow Y\overline{YZ(XY)}$ }\\
\ & \ & \mbox{with $X=C$, $Y=c$ and $Z=\tilde{\gamma}$.}\\
\hline
\vspace{-0.4cm }
\ & \ & \ \\
ABc\tilde{\gamma} & \overline{c\tilde{\gamma}} \longrightarrow \tilde{\gamma} &\mbox{By the rule $\overline{XY} \longrightarrow Y$ with $X=c$ and $Y=\tilde{\gamma}$.} \\

\end{array}
\]
In a previous example in
Subsection~\ref{Subsection:TheExperiment}, we showed that the
string $ABC\gamma \in \tilde \Sigma$ is not consistent with the
experiment. It is not hard to prove that in our grammar there
is no derivation of this string either.

We note and observe that the grammar has various desirable
properties. As such, the grammar is monotone\footnote{That is,
the length of each subsequent string in a derivation is at
least as long as the length of the previous one.}, where each
generated string contains at most one generating symbol.
Moreover, it is easily seen that once a string generated by the
grammar contains a composite generating symbol (like
$\overline{XY}$), each subsequently generated string will also
contain a composite generating symbol if it contains any
generating symbols at all. Indeed, the grammar is very simple.

\begin{theorem}\label{Theorem:GrammarAndExperimentCoincide}
The language $\mathcal{L}$ as defined in
Definition~\ref{definition:regularLanguageNineElements}
coincides with the set of consistent measurements as defined in
Definition~\ref{definition:experimentNineElements}.
\end{theorem}

\begin{proof}
We must show that, on the one hand any string in $\mathcal{L}$
is consistent with the experiment, and on the other hand, any
sequence of measurements consistent with the experiment is
derivable in $\mathcal{L}$.

The first implication is proven by an induction on the length
of the sequence of measurements. The fact that not measuring at
all is consistent is reflected by $\lambda \in \mathcal{L}$.
For non-empty sequence we distinguish between a context being
determined or not. The first case is covered by rules of the
form $\overline{XY} \to \mbox{righthand side}$. Note that in
the experiment, once there is a context fully determined, in
any future measurement there will be a (possibly different)
context that is fully determined. This is reflected in the
grammar in the sense, that once a composite generating symbol
(of the form $\overline{XY}$) enters the derivation, in all
subsequent derivations all generating symbols are composite. In
this sense, the composite generating symbols correspond exactly
to the case where there is a context fully determined. Note
that for any consistent new measurement there is a
corresponding rule (righthand side).
The second case is easier and corresponds to all rules of the form $\overline{X} \to \mbox{righthand side}$.\ \\
\medskip

A proof of the second implication proceeds by induction on the
length of a derivation. We first note that any non-empty string in
$\tilde \Sigma^*$ that is derivable in $\mathcal{L}$ is of the
form $\sigma \star \Xi$ where $\sigma \in \tilde \Sigma^*$ and $\Xi
\in \mathcal{G}$. Although each $\Xi \in \mathcal{G}$ is a
single separate symbol we already suggestively had composite
notations like $\overline{XY}$ for them. We define $l(\Xi)$ to
be ``the last symbol'' of $\Xi$ so that $l(\overline S)=
\lambda$, $l(\overline{X})= X$, and $l(\overline{XY})=Y$. We
now can prove by an easy induction on the length of a
derivation in $\mathcal{L}$, that for any string $\sigma \Xi$
that is derivable, the corresponding string of measurements
$\sigma l(\Xi)$ is consistent with the experiment. Again we use
here the distinction between composite generating symbols and
non-composite generating symbols and their correspondence to a
context being fully determined or not.
\end{proof}

Note that by the mere syntactic properties of the definition of
$\mathcal{L}$ we see that $\mathcal{L}$ is indeed a regular
language.
\begin{corollary}\label{Corollary:PMisRegular}
The set of consistent sequences of measurements of the
Peres-Mermin experiment is a regular language.
\end{corollary}
Regular languages are of Type 3 in the Chomsky hierarchy. As
such we have access to a corpus of existing theory. In
particular, there exists a method to determine the minimal
amount of states in a Deterministic Finite State Automata that
will accept $\mathcal{L}$. Moreover, we also have access to the
following proposition \cite{Flajolet2009}.

\begin{proposition}\label{Proposition:CountingNumberOfStrings}
Let $\mathcal{L}$ be a regular language. Then, there exist
polynomials $p_1,\ldots, p_k$ and constants $\lambda_1, \ldots,
\lambda_k$ such that the number of strings of length $n$ in
$\mathcal{L}$ is given by
\[
p_1(n) \lambda_1^n + \ldots + p_k(n) \lambda_k^n.
\]
\end{proposition}

As we have seen already, now that we have access to a smooth
inductive definition of $\mathcal{L}$ and thus of the set of
possible measurements according to the experiment described in
Definition~\ref{definition:experimentNineElements}, various
properties are readily proved using induction on the length of
a derivation in $\mathcal{L}$.

\begin{definition}\label{Definition:Agree}
We say that a string $\sigma\in \tilde\Sigma^*$ {\bf
determines} some observable $s\in \Sigma$ with value $v$,
whenever the sequence of measurements  corresponding to
$\sigma$ determines $s$ as defined in
Definition~\ref{definition:experimentNineElements} with value
$v$. We say that a string $\sigma \in \tilde \Sigma^*$
determines some context $c$, if $\sigma$ determines each
observable in $c$. We say that two strings $\sigma, \sigma' \in
\tilde \Sigma^*$ {\bf agree} on $s \in \Sigma$, whenever either
both do not determine $s$ or both determine $s$ with the same
value.
\end{definition}

With this definition at hand, we explicitly re-state some
observations that were already used in the proof of
Theorem~\ref{Theorem:GrammarAndExperimentCoincide}.

\begin{lemma}\label{lemma:AlwaysOneContextDetermined}
If some $\sigma \in \tilde \Sigma^*$ determines a context, then
any extension/continuation $\sigma\star\tau$ of $\sigma$ also
defines a context.
\end{lemma}

This lemma does not scale to systems of more qubits. The following lemma does.

\begin{lemma}\label{lemma:AtmostOneContextDetermined}
Each $\sigma \in \tilde \Sigma^*$ determines at most one context.
\end{lemma}


\section{Hidden variables and de-randomization}\label{Section:HiddenVariablesAndDerandomization}


In this section we wish to study the cost of various forms of
de-randomization of our experiment. First, the rather evident
observation is made that a finite set of hidden finite valued
variables can never account for indeterminism in an ideally
isolated repeatable experiment. Later we shall elaborate on
some refinements of this statement.


\subsection{Need for infinitely many hidden discrete variables}\label{subSection:HiddenVariablesAndDerandomization}


Let us reconsider our experiment from
Section~\ref{Section:ExperimentAndLanguage} where we measure
our sequence of observables. There is an amount of randomness
present. If we measure $A$ for the first time, the outcome can
be both $1$ or $-1$. However, it might be the case that the
particular outcome is dependent on some additional parameters
that we can not directly observe. We shall refer to those
additional parameters as \emph{hidden variables} (HVs).

The question now \adan{raises}, can we find a model that uses a
finite amount of hidden variables such that, once they have
been assigned some initial values at the outset of our
experiment, then the outcome of the experiment becomes fully
deterministic. We shall first see that the answer is NO if our
hidden variables can only take on a finite amount of values and
under some additional general assumptions. Later we shall
fine-tune on this result.

Let us first formulate these two general additional assumptions
that from now on shall be implicitly assumed throughout the
rest of this section.

\begin{quote}{\bf Assumption of independence of exact time}
We assume that in our experiment, the exact outcome of our measurement
does not depend on the exact time the measurement has taken place.
The only important thing is the history of measurements so far. Thus,
for any positive $t_1$ and $t_2$, measuring $B$ after $t_1$ seconds
after measuring $A$ should yield the same value as $B$ after $t_2$
seconds after measuring $A$, Ceteris Paribus.
\end{quote}

Moreover, we shall assume that
there is no relevant
interaction between our experiment and the outside world. Thus,
with these two assumptions, a deterministic explanation of our
experiment will exist of a function
\[
f(\vec{x},A):\ \  \ \ H \times \Sigma' \ \ \to H,
\]
where $\Sigma'$ is the set of observables, in our case $\Sigma'
= \Sigma = \{ A,B,C,a,b,c,\alpha,\beta,\gamma\}$ and $H$ is
some space where the (hidden) variables attain their values.
Thus basically, the function $f$ will tell you, once you know
all the underlying hidden variables in $\vec{x}$, what the next
value will be. So, the idea is that the hidden variables fully
determine the situation and the outcome of future measurements.
Thus, for convenience, and without loss of generality we may
assume that $\Sigma \subseteq H$.

We can easily extend $f$ in such a way that given the initial
situation in terms of the HVs we can obtain a full description,
again in terms of the HVs, of any future situation given that
we decide to perform a sequence of measurements. If $\sigma$ is
some sequence of observables, that is, $\sigma \in  \Sigma^*$,
we will denote by $f^{\sigma}(\vec{x})$ the values of the HVs
after performing this sequence $\sigma$ of measurements when we
start with initial conditions $\vec{x}$. Formally, we define
this as follows.

\begin{definition}[$f^{\sigma}$]
Let $f$ be a deterministic explanation of our experiment and
let $\sigma \in (\Sigma')*$, that is, $\sigma$ is a string of
observables. The function $f^{\sigma}$ is defined inductively
on the length of the string $\sigma$ as follows:
\[
\begin{array}{lll}
f^{\lambda} &:= \mathbb{I} \ \ &\mbox{ where $\mathbb{I}$ is the identity operator;}\\
f^{\tau \star A} (\vec{x}) &:= f(f^{\tau}(\vec{x}),A) \ \ \ \  &\mbox{ where $\star$ is the concatenation operator}.
\end{array}
\]
\end{definition}

Now basically, if $\vec{x}$ takes its values in some discrete
and finite space $H$, then there are only finitely many
starting conditions $\vec{x}$. With these finitely many degrees
of freedom in the starting conditions, we can never account for
the infinitude of choices that can be made in our sequence of
measurements. We can make this idea more formal in the
following easy theorem.

\begin{theorem}\label{theorem:noFiniteNumberOfHVs}
Let $f$ be a deterministic explanation of our experiment and
let $H$ be the space values where the hidden variables take on
their values. If $H$ is finite, then it can never account for
all the possible outcomes of our experiment. That is, in this
case, the hidden variables do not fully determine the full
course of measurements to follow.
\end{theorem}

\begin{proof}
The basic idea is that, if $H$ is finite, then $\vec{x}$ only
can be a finite number of different initial conditions. A
different sequence of measurements can only occur when the
initial values of $\vec{x}$ were different. Thus, for example,
the outcome $A = 1; b=1$ must necessarily have started with a
different initial condition $\vec{x}$  than the sequence $A =
1; b =-1$. As there are infinitely many different sequences
that are in our language $\mathcal{L}$, they can only be
accounted for by infinitely many different initial conditions.
\end{proof}


\subsection{Refinements and
lower-bounds}\label{SectionHolevoBound}


Now that we have seen in
Theorem~\ref{theorem:noFiniteNumberOfHVs} that no finite amount
of HVs suffices to de-randomize our experiment, we can ask
ourselves the following questions concerning refinements of
Theorem~\ref{theorem:noFiniteNumberOfHVs}.

\begin{enumerate}
\item\label{Item:HVperLength} What is the amount of binary
    HVs needed to explain all experiments of length $n$.
    That is, how many bits of memory are needed at least to
    explain from the initial conditions the full outcome of
    all possible sequences of measurements up to length
    $n$?

\item\label{Item:MemoryForRecognizing} What is the minimal
    amount of binary memory needed to \emph{recognize} any
    string in $\mathcal{L}$ by walking linearly through
    that string? Notice, in this question the aim is not to
    de-randomize the experiment by using HVs, rather it
    asks how much memory is needed in order to walk through
    a string of the form $A\tilde{b}Ab\gamma C\tilde{c}$
    without necessarily copying the entire string in some
    memory, and tell in the end\footnote{So this particular
    string $A\tilde{b}Ab\gamma C\tilde{c}$ is in
    $\mathcal{L}$.} of it whether or not the string is in
    $\mathcal{L}$.

\end{enumerate}
Question \ref{Item:MemoryForRecognizing} and a variation
thereof are addressed in the next section. An answer to
Question \ref{Item:HVperLength} essentially amounts to counting
the number of strings up to length $n$ and then taking the
logarithm of that number. Note that, in the light of
Proposition \ref{Proposition:CountingNumberOfStrings} and
Corollary \ref{Corollary:PMisRegular} we know the order of
magnitude of this number. That is, we know that the order of
magnitude of number strings of length $n$ in $\mathcal{L}$ is
(disregarding the polynomials, which are \adan{negligible} on a
logarithmic scale) $\lambda^n$, whence the order of strings up
to length $n$ is about $\lambda^{n+1}$ and the logarithm of
that results in a growth linear in $n$. Note that the total
number of strings of length up to $n$ in $\tilde \Sigma^*$ is
of order $|\Sigma^*|^{n+1}$ so that the logarithm of that is
also of order linear in $n$. In this sense, the number of HVs
needed to predict all experiments up to length $n$ is as bad
(high) as it could possibly be.


\section{Hidden variables and the Holevo bound}


Holevo \cite{Holevo} showed that the maximum information carrying capacity of a qubit is one bit. Therefore, a machine which simulates qubits but has a density of memory (in bits per qubit) larger than one violates the Holevo bound. In this section we show that a very broad class of machines that simulate the Peres-Mermin square violate the Holevo bound.

In \cite{MemoryCostQContextuality} deterministic automata are
presented that generate a subsets of $\mathcal{L}$. In that
paper lower bounds on the amount of states of these automata
are presented. Of course, as  $\mathcal{L}$ inhibits a genuine
amount of non-determinism any deterministic automata will
generate only a proper subset of $\mathcal{L}$ but never the
whole set $\mathcal{L}$ itself.

In the next subsection we shall introduce the notion of an
Memory-factoring Abstract Generating Automata (MAGA) for
$\mathcal{L}$ and prove a
lower bound on the number of states any MAGA should have if it
were to generate $\mathcal{L}$. In a sense, a MAGA for
$\mathcal{L}$ will generate all of $\mathcal{L}$.


\subsection{Memory-factoring Abstract Generating Automata}


In this project we are not interested in the details of
(abstract) machine `hardware'. Thus, in our definition of a
MAGA we will try to abstract away from the implementation
details of language generating automata. We will do this by
invoking the notion of \emph{computability}. By the
Church-Turing thesis (see, a.o. \cite{Sipser1997}) any
sufficiently strong and mechanizable model of computation can
generate the same set of languages, or equivalently, solve the
same same set of problems. Thus, instead of fixing one
particular model of computation and speak of computability
therein, we may just as well directly speak of computable
outright leaving the exact details of the model unspecified.

Basically, a MAGA $\mathcal{M}$ for $\mathcal{L}$ is an
abstract machine that will predict the outcome of a measurement
of some observable $s\in \Sigma$ in the experiment as described
in Definition~\ref{definition:experimentNineElements} given
that a sequence of measurements $\sigma \in \tilde \Sigma^*$
has already been done. If according to the experiment $s$ is
determined by $\sigma$ with value $v\in \{1,-1\}$, then
$\mathcal{M}$ should output $v$. If the observable $s$ is not
determined by $\sigma$ then $\mathcal{M}$ should output $r$
indicating that the experiment can randomly output a $-1$ or a
$1$.

The only requirement that we impose on a MAGA is that its
calculation in a sense \emph{factors through} a set of memory
states\footnote{A memory state is something entirely different
from, and hence is not to be confused with, a quantum state.}
in the sense that before outputting the final value, the
outcome of the calculation is in whatever way reflected in the
internal memory of the machine. Let us now formulate the formal
definition of a MAGA for $\mathcal{L}$.

\begin{definition}[MAGA]
A \emph{Memory-factoring Abstract Generating Automata} (MAGA)
for $\mathcal{L}$ is a quadruple $\langle M, M_0, M_1, S
\rangle$ with
\begin{enumerate}

\item
$S$ is (finitely or infinitely) countable set of memory states;

\item
all of $M$, $M_0$ and $M_1$ are computable functions such that

\begin{enumerate}

\item
$M = M_1 \circ M_0$;

\item
$M_0 \ :\ \mathcal{L} \times \Sigma \ \to \ S \times\Sigma$,\\
where\footnote{Here $\Pi_2$ is the so-called projection
function that projects on the second coordinate: $\Pi_2
(\langle x , y\rangle) = y$. Basically, $\Pi_2 \circ
M_0 = \mathbb{I}$ just says that $M_0$ only tells us
which state is defined by a sequence $\sigma \in \tilde
\Sigma^*$.} $\Pi_2 \circ M_0 = \mathbb{I}$;

\item
and $M_1\ : \ S \times \Sigma \ \to \ \{1,-1,r \}$,
\end{enumerate}

\end{enumerate}

such that

\[
\begin{array}{lrllr}
M(\sigma,s) =& 1 \ \ &\mbox{if  $\sigma$ determines $s$ with value}&\ &1;\\
M(\sigma,s) =& -1 \ \ &\mbox{``}&\mbox{"}&-1;\\
M(\sigma,s) =& r \ \ &\mbox{if  $\sigma$ does not determine $s$.}&\ &\ \\
\end{array}
\]
\end{definition}
In our definition, we have that $M_0 \ :\ \tilde \Sigma^{*}
\times \Sigma \ \to \ S \times\Sigma$, where $M_0$ does nothing
at all on the second coordinate, that is, on the $\Sigma$ part.
We have decided to nevertheless take the second coordinate
along so that we can easily compose $M_0$ and $M_1$ to obtain
$M$. This is just a technical detail. The important issue is
that the computation \emph{factors} through the memory. That
is, essentially we have that $M:  \mathcal{L}\times \Sigma
\stackrel{M_0}{\longrightarrow} S
\stackrel{M_1}{\longrightarrow} \{ 1,-1,r\}$ where `$M_1$
borrows some extra information on the $\Sigma$-part of the
original input'.

\begin{theorem}
The class of MAGA-computable functions is the full class of computable functions.
\end{theorem}

\begin{proof}
It is easy to see that if we have infinite memory, we can
conceive any Turing Machine $\mu$ with the required
input-output specifications as a MAGA where $M_0$ is just the
identity, $S$ is coded by the tape input, and $M_1$ is the
function computed by $\mu$. Thus, the class of MAGA-computable
functions is indeed the full class of computable functions.
\end{proof}


\subsection{A lower bound for the Peres-Mermin square}


With the formal definition at hand we can now state and proof
the main theorem for the Peres-Mermin square. In the proof we
will use the so-called Pigeon Hole Principle (PHP). The PHP
basically says that there is no injection of a finite set into
a proper subset of that set. Actually we will only use a
specific case of that which can be rephrased as, if we stuck
$n+1$ many pigeons in $n$ many holes, then there will be at
least one hole that contains at least two pigeons.

\begin{theorem}\label{theorem:lowerBoundContextsTimesFreedom}
Any MAGA for $\mathcal{L}$ contains at least 24 states. That is, if $\langle M,M_0 ,M_1 ,S \rangle$ is such a MAGA, then $|S| \geq 24$.
\end{theorem}

\begin{proof}
%
%
We shall actually use a slightly modified version of a MAGA to prove our theorem. In the unmodified MAGA, the function $M_0$ tells us what state is attained on what sequence of measurements $\sigma \in \tilde \Sigma^*$. In the modified MAGA we will only require that $M_0$ will tell us in what state the machine is whenever $\sigma$ determines a full context.

To express this formally we define
\[
\tilde \Sigma^+ \ := \ \{ \sigma \in \tilde \Sigma^* \mid \mbox{ $\sigma$ determines a full context of observables}\}.
\]
Thus, instead of requiring that $M_0$ maps from $\tilde
\Sigma^* \times \Sigma$ to $S\times \Sigma$, we will require
that $M_0$ maps from $\tilde \Sigma^+ \times \Sigma$ to
$S\times \Sigma$. Clearly, if we have a lower bound for any
MAGA $\mathcal{M}$ with this restriction on $M_0$, we
automatically have the same lower bound for any MAGA
$\mathcal{M}'$ outright. This is so as any MAGA $\mathcal{M}'$
trivially defines a restricted MAGA $\mathcal{M}$ by just
restricting the domain of $M_0$ to $\tilde \Sigma^+ \times
\Sigma$.
\medskip

To continue our proof, let $s_1, \ldots, s_{24}$ enumerate all
possible combinations $\langle c, v_1,v_2\rangle$ of contexts
and the first two\footnote{For horizontal contexts we will
enumerate from left to right and for \adan{vertical} contexts
we will enumerate from top to bottom. Thus, for example, the
first two observables of the context $\langle C, c,
\gamma\rangle$ are $C$ and $c$.} values of the first two
observables of that context $c$. Note that there are indeed $6
\times 2\times 2 = 24$ many such combinations. We define a map
\[
C \ : \ \ \ \tilde \Sigma^+ \ \to \ \{ s_1, \ldots, s_{24} \}
\]
in the canonical way, mapping an element $\sigma \in \tilde
\Sigma^+$ to that $s_i$ that corresponds to the triple
consisting of the context that is determined by $\sigma$
followed by the first two values of the the first two
observables of that context.  By Lemmas
\ref{lemma:AlwaysOneContextDetermined} and
\ref{lemma:AtmostOneContextDetermined} the function $C$ is
well-defined.

Now, let $\sigma_1,\ldots,\sigma_{24}$ be representatives in
$\tilde \Sigma^+$ of $s_1,\ldots,s_{24}$ such that
$C(\sigma_i)= s_i$. For a contradiction, let us assume that
there exists some restricted MAGA $\langle M, M_0, M_1,
S\rangle$ for $\mathcal{L}$ with $|S| \leq 23$. By the Pigeon
Hole Principle, we can choose for this MAGA some $\sigma_i$ and
some different $\sigma_j$ such that\footnote{Par abus de
langage we will write $M_0(\sigma_i)$ as short for
$\Pi_2(M_0(\sigma_i,s))$.}
\[
M_0(\sigma_i) = M_0(\sigma_j).
\]
We now use the following claim that shall be proved below.
Recall from Definition \ref{Definition:Agree} what it means for
two sequences to agree on some variable.
\begin{claim}
If $\sigma_k \neq \sigma_l$ then there is some $s \in \Sigma$ such that $\sigma_k$ and $\sigma_l$ disagree on $s$.
\end{claim}
Once we know this claim to hold it is easy to conclude the
proof. Consider any $s\in \Sigma$ on which $\sigma_i$ and
$\sigma_j$ disagree. By the definition of $M$ we should have
that $M(\sigma_i,s) \neq M(\sigma_j,s)$. However as
$M_0(\sigma_i,s) = M_0(\sigma_j,s)$ we see that $M_1\circ M_0
(\sigma_i,s) = M_1\circ M_0 (\sigma_j,s)$. But $M_1\circ M_0 =
M$, which contradicts $M(\sigma_i,s) \neq M(\sigma_j,s)$. We
conclude that $M$ can not have 23 or less states.
\medskip

Thus to finalize our proof we prove the claim. Let $C(\sigma_k)
= \langle c^k, v_0^k,v_1^k\rangle \neq  \langle c^l,
v_0^l,v_1^l \rangle = C(\sigma_l)$. If $c^k \neq c^l$ then
$c^k$ contains at least two observables on which $\sigma_k$ and
$\sigma_l$ agree as each of these $\sigma$'s only determine
observables in their respective contexts.

In case $c^k=c^l$, then one of $v_0^k, v_1^k$ differs from the
corresponding one in $v_0^l, v_1^l$ giving rise to a
disagreement between $\sigma_k$ and $\sigma_l$.

This concludes the proof of the claim and thereby of
Theorem~\ref{theorem:lowerBoundContextsTimesFreedom}.
\end{proof}

\begin{remark}
Note that the proof of
Theorem~\ref{theorem:lowerBoundContextsTimesFreedom} nowhere
invokes the notion of computability therefore proving actually
something stronger.
\end{remark}

One can easily see that for $\mathcal{L}^+$ the obtained lower
bound is actually sharp in the sense that there is a MAGA with
24 memory states for $\mathcal{L}^+$. However, is seems that
for $\mathcal{L}$ this is not the case.

Any memory-factoring device that \emph{recognizes}
$\mathcal{L}$ --let us call that a MARA for convenience-- can
be turned into a MAGA for $\mathcal{L}$. In analogy with our
MAGA, we conceive such a device as some $M' : \mathcal{L}
\times \tilde \Sigma \stackrel{M'_0}{\longrightarrow} S' \times
\Sigma  \stackrel{M'_1}{\longrightarrow} \{ \mbox{YES},
\mbox{NO}\}$. Now, any such device can be transformed to a MAGA
$\langle M, M_0, M_1, S \rangle$ with $S = S' \times S'$ and
$M_0$ maps any pair $\langle \sigma, s\rangle$ to the pair
$\langle M_0'(\sigma, s), M_0'(\sigma, \overline{s})\rangle$.
Then $M_1$ will output $r$ only if both $\sigma s$ and $\sigma
\overline{s}$ are in $\mathcal{L}$ and do the obvious thing
otherwise. To put it more formal, $M_1$ acts on a pair from
$S'\times S'$ using $M_1'$ for each coordinate and consequently
mapping (YES, YES) to $r$, mapping (YES, NO) to 1,
and\footnote{Note that (NO, NO) cannot occur.} mapping (NO,
YES) to $-1$. Note that if $n$ is a lower bound for a MAGA,
then the lower bound on a MARA that corresponds to this
construction is of size $\sqrt n$ thus giving a partial answer
to Question \ref{Item:MemoryForRecognizing} from
Section~\ref{SectionHolevoBound}.


\subsection{Scaling}


We first note that the proof of
Theorem~\ref{theorem:lowerBoundContextsTimesFreedom} is very
amenable to generalizations:
\begin{remark}\label{Remark:GeneralizationLowerBound}
The proof of
Theorem~\ref{theorem:lowerBoundContextsTimesFreedom} easily
generalizes under some rather weak conditions giving rise to
lower bounds of $\# \mbox{contexts} \times 2^{(\#\mbox{degrees
of freedom in one context})}$.
\end{remark}
The Peres-Mermin square $(\dag)$ that corresponds to the
two-qubit system can be generalized in various ways
\adan{\cite{Cabello10,Cabello11}}. One particular
generalization for $n$ qubits gives rise to a system where each
context consists of exactly $d$ elements with $d= 2^n$
\adan{\cite{Cabello11}}. Moreover, there are $c:= \prod_{k=1}^m
(2^k +1)$ different many such contexts. Each context is now
determined by a particular selection of $n$ of its elements. We
shall denote the corresponding languages by $\mathcal{L}_n$.
Thus, what we have called $\mathcal{L}$ so far in this paper,
would correspond to $\mathcal{L}_2$.

\begin{theorem}\label{Theorem:GeneralizedLowerBound}
Any MAGA for $\mathcal{L}_n$ contains at least $2^n \cdot \prod_{k=1}^n (2^k +1)$ many different memory states.
\end{theorem}

\begin{proof}
Basically this is just by plugging in the details of the
languages $\mathcal{L}_n$ into remark
\ref{Remark:GeneralizationLowerBound}. Let us very briefly note
some differences with the proof of
Theorem~\ref{theorem:lowerBoundContextsTimesFreedom}. The main
difference is that the set $\mathcal{L}_2^+$ is nice: once a
string is in there, any extension is as well. However, this
does not impose us to, again define a restricted MAGA by
restricting the domain of $M_0$ to $\mathcal{L}_n^+$. Again, we
consider the $n+1$-tuples consisting of a context with some
values for the\footnote{For each context, we fix some $n$
observables that determine that context as these are not
uniquely defined} $n$ observables that determine this context
and choose some correspondence between these tuples and some
representing sequences $\sigma_i \in \tilde \sigma_n^*$.
Clearly, these $\sigma_i \in \mathcal{L}_n^+$. As the claim
obviously holds also in the general setting, the assumption
that the memory states $s_1, \ldots, s_{2^n \cdot \prod_{k=1}^n
(2^k +1)-1}$ suffice yields together with the PHP to a
contradiction as before.
\end{proof}

As was done in \cite{Cabello2010} and in
\cite{MemoryCostQContextuality}, we can consider the
information density $d_n$ for the corresponding languages
defined as the number of classical bits of memory needed to
simulate a qubit:
\[
d_n:= \frac{\log_2(|S_n|)}{n}.\ \ \ \ \ \ \ (+)
\]
If we apply the lower bound for $S_n$ --the number of memory
states for a MAGA for $\mathcal{L}_n$-- from
Theorem~\ref{Theorem:GeneralizedLowerBound} to $(+)$ we obtain
\[
\begin{array}{lll}
S_n \ \ \ & \geq \ \ \ \ & 2^n \cdot \prod_{k=1}^n (2^k +1)\\
\ & \geq & 2^n \cdot \prod_{k=1}^n (2^k)\\
\ & \geq & 2^n \cdot  (2^{\sum_{k=1}^nk})\\
\ & \geq & 2^n \cdot  (2^{\frac{n(n+1)}{2}})\\
\end{array}
\]
whence $d_n$ is approximated (from below) in the limit by
$\frac{\log_2(2^n \cdot  (2^{\frac{n(n+1)}{2}}))}{n} =
\frac{n+3}{2} \sim \frac{n}{2}$. Thus, the information density
in \adan{this generalization of the Peres-Mermin square} grows
linear in the number of qubits. However, as observed before, any density more
than~1 implies a violation of the Holevo bound.


\end{document}